\documentclass[11pt,times]{article}
\usepackage{amsfonts}
\usepackage{amssymb,amsmath,amsthm}
\usepackage{amssymb, amsthm, algorithm, algpseudocode, amsmath}
\usepackage{epsfig}
\usepackage{latexsym}
\usepackage{enumerate}
\usepackage{fullpage,setspace}
\usepackage[dvips, paper=letterpaper, top=0.95in, bottom=.75in, left=0.9in, right=0.95in, nohead, includefoot, footskip=.25in]{geometry}
\usepackage{bbm}
\usepackage{stmaryrd}
\usepackage[usenames,dvipsnames]{color}
\usepackage{url}
\usepackage{comment}

\usepackage[numbers,sort]{natbib}

\usepackage[dvipsnames,usenames]{xcolor}
\usepackage[letterpaper=true,colorlinks=true,pdfpagemode=none,urlcolor=RoyalBlue,linkcolor=RoyalBlue,citecolor=OliveGreen,pdfstartview=FitH]{hyperref}
\usepackage{prettyref}

\let\pref=\prettyref
\newcommand{\savehyperref}[2]{\texorpdfstring{\hyperref[#1]{#2}}{#2}}

\newrefformat{eq}{\savehyperref{#1}{\textup{(\ref*{#1})}}}
\newrefformat{chp}{\savehyperref{#1}{Chapter~\ref*{#1}}}
\newrefformat{lem}{\savehyperref{#1}{Lemma~\ref*{#1}}}
\newrefformat{def}{\savehyperref{#1}{Definition~\ref*{#1}}}
\newrefformat{thm}{\savehyperref{#1}{Theorem~\ref*{#1}}}
\newrefformat{cor}{\savehyperref{#1}{Corollary~\ref*{#1}}}
\newrefformat{cha}{\savehyperref{#1}{Chapter~\ref*{#1}}}
\newrefformat{sec}{\savehyperref{#1}{Section~\ref*{#1}}}
\newrefformat{app}{\savehyperref{#1}{Appendix~\ref*{#1}}}
\newrefformat{tab}{\savehyperref{#1}{Table~\ref*{#1}}}
\newrefformat{fig}{\savehyperref{#1}{Figure~\ref*{#1}}}
\newrefformat{hyp}{\savehyperref{#1}{Hypothesis~\ref*{#1}}}
\newrefformat{alg}{\savehyperref{#1}{Algorithm~\ref*{#1}}}
\newrefformat{item}{\savehyperref{#1}{Item~\ref*{#1}}}
\newrefformat{step}{\savehyperref{#1}{step~\ref*{#1}}}
\newrefformat{conj}{\savehyperref{#1}{Conjecture~\ref*{#1}}}
\newrefformat{fact}{\savehyperref{#1}{Fact~\ref*{#1}}}
\newrefformat{prop}{\savehyperref{#1}{Proposition~\ref*{#1}}}
\newrefformat{claim}{\savehyperref{#1}{Claim~\ref*{#1}}}

\newtheorem{theorem}{Theorem}
\newtheorem{lemma}[theorem]{Lemma}
\newtheorem{corollary}[theorem]{Corollary}

\newcommand{\SSS}{\mathcal S}
\newcommand{\MMM}{\mathcal M}

\newcommand{\argmax}{\operatorname{argmax}}

\newcommand{\alg}{\operatorname{ALG}}
\newcommand{\opt}{\operatorname{OPT}}

\newcommand{\Ex}{\mathbb E}

\newif\ifcomments
\commentstrue
\ifcomments

\else
\newcommand{\remove}[1]{}
\fi

\begin{document}
\title{Welfare Maximization with Production Costs:\\ A Primal Dual Approach}

\author {Zhiyi Huang\thanks{This work was done when the author was at Stanford University. Supported in part by ONR PECASE Grant N000140910967.}\\
The University of Hong Kong\\
zhiyi@cs.hku.hk
\and
Anthony Kim\thanks{Supported in part by an NSF Graduate Research Fellowship and NSF grant CCF-1215965.}\\ 
Stanford University\\
tonyekim@stanford.edu
}

\date{}


\maketitle
\thispagestyle{empty}

\begin{abstract}
\small\baselineskip=9pt
We study online combinatorial auctions with production costs proposed by \citet{BGMS11} using the online primal dual framework. In this model, buyers arrive online, and the seller can produce multiple copies of each item subject to a non-decreasing marginal cost per copy. The goal is to allocate items to maximize social welfare less total production cost. For arbitrary (strictly convex and differentiable) production cost functions, we characterize the optimal competitive ratio achievable by online mechanisms/algorithms. We show that online posted pricing mechanisms, which are incentive compatible, can achieve competitive ratios arbitrarily close to the optimal, and construct lower bound instances on which no online algorithms, not necessarily incentive compatible, can do better. Our positive results improve or match the results in several previous work, e.g., \citet{BGN03}, \citet{BGMS11}, and \citet{BG13}. Our lower bounds apply to randomized algorithms and resolve an open problem by \citet{BG13}.
\end{abstract}


\newpage

\section{Introduction}\label{sec:intro}

Consider a seller with $m$ heterogeneous items to allocate to $n$ heterogeneous buyers to maximize social welfare, that is, the sum of the buyers' value for the items they obtain. When buyers have combinatorial value functions over bundles of items, it is known as combinatorial auctions, a central problem in economics and algorithmic game theory.

Combinatorial auctions are computationally challenging. 
For instance, when buyers have arbitrary value functions, there are no polynomial-time algorithms that approximately maximize social welfare within a factor better than $m^{1/2}$ (e.g., \cite{BN07}).
Moreover, if we take the strategic behavior of self-interested buyers into account, i.e., focusing on polynomial-time and incentive compatible mechanisms, then even computationally simple special cases become intractable.
For example, there are polynomial-time $\tfrac{e}{e-1}$-approximation algorithms when buyers have submodular value functions \cite{Von08}, but no polynomial-time and incentive compatible mechanisms can be better than an $m^\gamma$-approximation for some constant $\gamma$ \cite{DV14}.

Part of the difficulty of combinatorial auctions comes from the stringent supply constraints -- the seller has only one copy per item. In many applications, the seller may have multiple copies of each item, or can even produce an arbitrary number of copies paying certain production cost. Therefore, it is natural to consider variants of combinatorial auctions with relaxed supply constraints, hoping that they are more tractable.

The first set of results along this line considers having logarithmically many copies of each item.
With $\Omega(\tfrac{1}{\epsilon^2} \log m)$ copies, an $(1 - \epsilon)$-approximation algorithm is folklore because the standard linear program relaxation of combinatorial auctions has an integrality gap of $1-\epsilon$.
\citet{LS11} further showed how to convert this algorithm into an $(1-\epsilon)$-approximate and incentive compatible mechanism.
\citet{BGN03} considered the online setting where buyers arrive online and the seller must allocates to each buyer at his arrival without any information about future buyers. 
Assuming the buyers' value for any bundle is in the range of $[v_{\min}, v_{\max}]$\footnote{An upper bound on values is necessary for any non-trivial competitive ratio. Otherwise, after the seller has exhausted the supply of an item, there could be a buyer with value for the item arbitrarily large relative to the previous buyers' values.} and there are $\Omega(\log (v_{\max}/v_{\min}))$ copies per item, they introduced an $O(\log (v_{\max}/v_{\min}))$-competitive\footnote{Some previous work assume knowing $v_{max}$ and, instead of $v_{\min}$, the number of buyers, $n$. In this case, using $v_{\max} / n$ as an effective $v_{\min}$ leads to $O(\log n)$-competitive algorithms. Most techniques can be translated between the two settings. We assume knowing $v_{\max}$ and $v_{\min}$ throughout this paper for consistency.} and incentive compatible mechanism.

Recently, \citet{BGMS11} studied online combinatorial auctions in a more general model with production costs.
In this model, the seller may produce any number of copies of the items while paying a non-decreasing marginal production cost per copy.
The goal is to maximize social welfare less the total production cost. 
The aforementioned model is a special case of the production cost model with the marginal production cost being a zero-infinity step function.
\citet{BGMS11} considered several simple marginal production cost functions, including linear, polynomial and logarithmic functions, and proposed constant competitive algorithms for these special cases. 
They also studied general cost functions assuming the values are between $v_{\min}$ and $v_{\max}$, and introduced logarithmic competitive algorithms. 
However, it is not clear if their competitive ratios are optimal, even for the special cases, and there was no characterization of the optimal competitive ratio achievable on a per cost function basis.

In this paper, we further investigate online combinatorial auctions with production costs. 
We seek to develop a unified framework that provides online mechanisms/algorithms with the optimal competitive ratios for arbitrary production costs using the online primal dual framework (see, e.g., \citet{BN09} for a comprehensive survey).
Informally, the online primal dual framework considers the linear program relaxation of an optimization problem and its dual program, and designs online algorithms based on the structure of the linear programs and complementary slackness.
It is not clear, however, how to model production costs using linear programs.
Instead, we use an extension of the technique to convex programs and Fenchel duality following the recent work of \citet{DJ12}, \citet{DH14}, etc.

\subsection{Our Contributions and Techniques} 

Our main contribution is a characterization of the optimal algorithms/mechanisms and their competitive ratios for online combinatorial auctions with arbitrary production costs via an online primal dual approach:

We start with a fractional version of the problem in \pref{sec:frac}, where there are infinitely many buyers each of which wants at most an infinitesimal amount of each item.
The fractional version allows us to focus on the online nature of the problem, while ignoring extra complications from the integrality gap of using convex program relaxations.
Then, we characterize a parameter $\alpha(f)$ that depends on the production cost function $f$, and show that (1) there are $(\alpha(f) + \epsilon)$-competitive and incentive compatible mechanisms, and (2) there are no $(\alpha(f) - \epsilon)$-competitive online algorithms, for arbitrarily small $\epsilon > 0$.
The optimal competitive ratio achievable $\alpha(f)$ is the infimum of parameters $\alpha$ such that a differential equation parameterized by $\alpha$ has a feasible solution.
More specifically, a feasible solution to the differential equation with $\alpha = \alpha(f) + \epsilon$ would yield a relation between the primal and dual variables in the online primal dual approach and, consequently, a competitive online mechanism. 
Our lower bound is obtained by constructing a family of instances such that if there is an online algorithm that is $(\alpha(f) - \epsilon)$-competitive for these instances, then there is a feasible solution to the differential equation parameterized by $\alpha = \alpha(f) - \epsilon$, contradicting our definition of $\alpha(f)$. 
In short, our mechanism and lower bound for a given production cost function $f$ reduce to the same differential equation and thus are optimal.
To the best of our knowledge, the only other work with this type of characterization of the optimal competitive ratio is the work by \citet{DJ12} on online matchings with concave returns.
Finally, we stress that our mechanisms are incentive compatible and their competitive ratios are optimal even when compared to non-incentive compatible algorithms.

In \pref{sec:int}, we study the integral version and show that under certain conditions, the optimal competitive ratios for the fractional case extend to the integral case with an arbitrarily small loss, say $\epsilon > 0$, in the competitive ratio.
In particular, we have a $((1+\frac{1}{d})^{d+1} d + \epsilon)$-competitive algorithm for the polynomial marginal production costs $c(y) = y^d$ and a $(4 + \epsilon)$-competitive algorithm for logarithmic marginal production costs. 
These results improve the competitive ratios previously achieved by \citet{BGMS11}. Our $(4 + \epsilon)$-competitive algorithm further applies to any concave marginal production costs.
We summarize these results in \pref{tab:comp}.

\begin{table*}
\centering
	\caption{Competitive ratios for various marginal production cost functions (for arbitrarily small $\epsilon$):}
	\smallskip
	\begin{tabular}{|l|c|c|}
  \hline 
   & \citet{BGMS11} & This Paper \\
  \hline 
  Linear, $c(y) = ay + b$ & 6 & $4(1+\epsilon)$ \\
  \hline
  Polynomial, $c(y) = a y^d$ & $(1+\epsilon) 4d$ & $(1+\epsilon) (1+\frac{1}{d})^{d+1} d$ \\
  \hline
  Logarithmic, $c(y) = \ln(1+y)$ & $4.93$ & $4(1+\epsilon)$ \\
  \hline
  \end{tabular}
  \label{tab:comp}
\end{table*}

Finally, we consider in \pref{sec:limsupply} the case when the buyers' values are between $v_{\min}$ and $v_{\max}$.
We use the same framework to derive nearly optimal incentive compatible mechanisms with a logarithmic competitive ratio for supply-$k$ online combinatorial auctions. 
As before, our mechanisms are derived from the same differential equation that characterizes online combinatorial auctions.
We also show an almost matching logarithmic lower bound that applies to randomized algorithms, resolving an open problem by \citet{BG13}.

\subsection{Other Related Work}

There is a vast literature on maximizing social welfare in combinatorial auctions. 
In addition to the related work we have already discussed, we refer readers to the survey by \citet{BN07} and the references therein. 
Also, we note recent positive results by \citet{DRY11} when buyers have coverage value functions or matroid rank sums value functions.

Our mechanisms fall into a family known as posted pricing mechanisms, where the seller posts item prices to each buyer and let the buyer pick his favorite bundle of items given the prices.
Posted pricing mechanisms are incentive compatible and are widely used for revenue maximization (see, e.g., \citet{BBM08}, \citet{CHK09}, \citet{CHMS10}).
Some of these results (e.g., \citet{CHK09}) also imply bounded approximation ratios for social welfare maximization.

All of the aforementioned results focus on the case when each item has only one copy, in which case strong assumptions on the value functions are needed to achieve non-trivial positive results.
On the other hand, we consider the production cost model by \citet{BGMS11}, i.e., allowing multiple copies of each item subject to production costs.  
As a result, we are able to achieve positive results for arbitrary value functions.

Our primal dual approach is related to the recent work by \citet{AGK12}, \citet{GKP12}, and \citet{Tha13} on the online scheduling problem using online primal dual analysis or dual fitting with Lagrangian duality. 
Lagrangian duality is defined even for non-convex programs and therefore can be applied to problems without a natural convex program relaxation.
In contrast, we use convex programs and Fenchel duality; fenchel duality is defined only for convex programs but generally presents richer structures that guide the design and analysis of online algorithms.



\section{Preliminaries}\label{sec:prelim}

Let there be a seller with a set of $m$ items, represented by $[m]$, and a group of $n$ buyers that arrive online. 
We use $i$ to represent indices of buyers, $j$ to represent indices of items, and $S$ to represent bundles.

Each item $j$ is associated with a production cost function $f_j: \mathbb{R}^+ \rightarrow \mathbb{R}^+$ where $f_j(y)$ is the total cost to produce $y$ units of item $j$, i.e., the $y$-th (integral) unit costs $f_j(y) - f_j(y-1)$. 
Our results will depend on certain technical properties of the cost functions which we will make clear in the theorem statements.
For simplicity of presentation, we assume all $f_j$ are identical and omit their subscripts in the rest of this paper. 
Our results extend to non-identical $f_j$'s as well, with a dependency on the ``worst'' $f_j$.

Let $\SSS$ be the set of available bundles that can be allocated to buyers. A bundle $S$ of items is represented by a vector $(a_{1S}, \ldots, a_{mS})$ where $a_{jS}$ is the number of units of item $j$ in the bundle. 
We assume $\SSS$ contains the empty bundle $\vec{0} = (0, \ldots, 0)$ and there is an upper bound $\Delta y$ on the maximal number of units of each item in a bundle. 
Readers may think of $\Delta y = 1$ and $\SSS = \{0, 1\}^m$ as a concrete example, which corresponds to the standard setting of combinatorial auctions. 
In general, $\Delta y$ can be any positive number and $\SSS$ can be any subset of $[0, \Delta y]^m$ containing $\vec{0}$.
For simplicity of exposition, we assume $\SSS = \{0, \Delta y\}^m$ and say $j \in S$ if $a_{jS} = \Delta y$. 
Our analysis extends to the general case.

Each buyer $i$ has a private value function $v_i : \SSS \mapsto \mathbb{R}^+$ where $v_i(S)$ is buyer $i$'s value for getting bundle $S \in \SSS$ of items.
We do not require any assumptions on the value functions and allow them to have arbitrary complement and substitute effects.

At the beginning, the seller does not have any information about the buyers, that is, $v_i$'s and even $n$ are not known to the seller. 
Upon the arrival of each buyer $i$, the buyer reports a value function $\hat{v}_i$ (which may or may not be his true value function $v_i$) and the seller (irrevocably) allocates a bundle $S_i \in \SSS$ to the buyer and charge a payment $P_i$, based on $\hat{v}_i$ and the reported values and allocations of previous buyers.

The resulting allocation rule along with the payment rule constitute an online mechanism.
Since the seller does not know the buyers' value functions upfront, he needs to incentivize them to truthfully report their value functions.
A mechanism is incentive compatible if each buyer $i$ maximizes the expectation of his utility, i.e., $v_i(S_i) - P_i$, by reporting $\hat{v}_i = v_i$.
If the buyers are not strategic, i.e., they always truthfully report their value functions, then we only need the allocation rule which is simply an online algorithm. 

The objective is to allocate items to maximize the expectation of social welfare, which is the sum of the buyers' value for the bundles they get, less the total production cost, i.e., $\sum_i v_i(S_i) - \sum_j f(y_j)$ where $y_j$ denotes the total amount of item $j$ that has been sold so far.
We measure the performance of mechanisms/algorithms under the standard competitive analysis framework. 
Let $W(M)$ denote the expected objective value of a mechanism $M$.
Let $\opt$ denote the optimal objective value in hindsight.
A mechanism $M$ is $\alpha$-competitive if there exists a constant $\beta$, independent of $n$ and $v_i$'s, such that
\[ W(M) \ge \tfrac{1}{\alpha} \opt - \beta \]
for all possible instances. 
Clearly, $\alpha$ is always at least $1$, and the closer to $1$ the better.
Our goal is to characterize the optimal competitive ratio $\alpha$ achievable by any online mechanisms/algorithms.

\paragraph{Posted Pricing Mechanisms}

A particularly related family of mechanisms are posted pricing mechanisms.
Upon a buyer's arrival, the seller chooses item prices $p_j$ and lets the buyer pick his utility-maximizing bundle, namely, $\argmax_{S \in \SSS} v_i(S) - \sum_j a_{jS} p_j$, breaking ties arbitrarily.
The mechanism may use different prices for different buyers.
Posted pricing mechanisms are incentive compatible and are widely used for revenue maximization (e.g., \cite{BBM08, CHMS10, CHK09}). 
Recently, \citet{BGMS11} extended the use of posted pricing mechanisms to online combinatorial auctions with production costs; in particular, the price of an item $j$ only depends on the amount of the item that has been sold.
Thus, their mechanisms can be represented by a {\em pricing function} $p : \mathbb{R}^+ \mapsto \mathbb{R}^+$ where $p(y)$ is the price per unit of item $j$ if $y$ units of the item has been sold.
In this paper, we characterize the optimal competitive ratio achievable by any online algorithms and show that the optimal competitive ratio can be achieved by the posted pricing mechanisms.
See below for a formal description of the {\em posted pricing mechanism} $\MMM_p$ defined by a pricing function $p$:

\begin{algorithm}[H]
\caption{$\MMM_p$, pricing mechanism with pricing function $p$}\label{alg:pricing}
\begin{algorithmic}[1]
\State Initialize $y_j = 0$ for all $j$
\For{$i=1, \ldots, n$}
	\State Offer item $j$ at price $p_j = p(y_j)$ for all $j$
	\State Buyer $i$ chooses bundle $S$ and pays $\sum_{j\in S} p_j$
	\State Update $y_j = y_j + \Delta y$ for all $j\in S$
\EndFor
\end{algorithmic}
\end{algorithm}


\paragraph{Online Primal Dual Algorithms}

While our mechanisms are posted pricing mechanisms, we did not commit to them a priori. 
Instead, we derive our mechanisms from a principled primal dual analysis.
Consider the following convex program relaxation ($P$) of our problem, on the top, and its Fenchel dual program ($D$) (see, e.g.,  \citet{Dev10} for more discussions of Fenchel duality and \pref{app:deriv} for the derivation of the dual program):
\begin{align*}
\textstyle \max_{x, y} ~~ & \textstyle \sum_i \sum_S v_{iS} x_{iS} - \sum_j f(y_j) \\
\forall i: \quad & \textstyle \sum_S x_{iS} \leq 1 \\
\forall j: \quad & \textstyle \sum_i \sum_S a_{jS} x_{iS} \leq y_j\\ 
& x, y \geq 0\\
\end{align*}
\begin{align*}
\textstyle \min_{u, p} ~~ & \textstyle \sum_i u_i + \sum_j f^*(p_j) \\
\forall i, S: \quad & \textstyle u_i + \sum_j a_{jS} p_j \geq v_{iS} \\
& u, p \ge 0 
\end{align*}

In the primal program, variable $x_{iS}$ indicates whether or not buyer $i$ purchases bundle $S$. 
Since we want to maximize the objective function and $f$ is an increasing function, we may assume without loss that $y_j = \sum_i \sum_S a_{jS} x_{iS}$, namely, the total number of units of item $j$ that have been allocated.

In the dual objective, $f^*(p) = \sup_{y \ge 0} \{ p y - f(y) \}$ is the convex conjugate of $f$.
When $f$ is strictly convex and differentiable, $f'(y)$ and ${f^*}'(p)$ are inverses of each other.
We interpret $p_j$ as the price per unit of item $j$ and $u_i$ as the utility of buyer $i$.

When $a_{jS}$ are binary and $f$ is a step function that equals $0$ for $y \in [0,1]$ and $\infty$ for $y>1$, these programs become the standard primal and dual linear program relaxations for the combinatorial auctions, without production costs and with one copy per item (see \pref{app:singleunit}).

Upon the arrival of buyer $i$, there is a new dual variable $u_i$ and a set of new dual constraints, $u_i \ge v_{iS} - \sum_j a_{jS} p_j$ for all $S$. To maintain dual feasibility while minimizing the increase of the dual objective, we let $u_i = \min_S v_{iS} - \sum_j a_{jS} p_j$; 
by complementary slackness, we also let $x_{iS} = 1$. 
This allocation rule corresponds to letting buyer $i$ pick his utility maximizing bundle with $u_i$ equaling his utility.
After the allocation, $y_j$ increases for each item $j$ in the allocated bundle.
Consequently, we need to adjust the corresponding dual variables $p_j$. 
In the offline optimal solution, $y_j$ and $p_j$ shall form a complementary pair, i.e., $p_j = f'(y_j)$. 
In our online problem, however, the algorithm does not know the final demand and, therefore, let $p_j$ be $f'$'s value at some estimated final demand; in general, we let $p_j(y_j)$ be a function of the current demand $y_j$.

Let $P^i$ and $D^i$ be the primal and dual objective values, respectively, after serving buyer $i$; 
let $P^0$ and $D^0$ be the values at initialization and $P^n$ and $D^n$ be the values at termination.
Throughout the process, we maintain a feasible primal solution $(x, y)$ and a feasible dual solution $(u, p)$. 
We use superscript $i$ to denote the current values of the primal and dual variables after serving buyer $i$ and before serving buyer $i+1$.
If we could show that $P^n \ge \frac{1}{\alpha} D^n - \beta$, then our mechanism $M$ is $\alpha$-competitive since $W(M) = P^n \ge \frac{1}{\alpha} D^n - \beta \ge \frac{1}{\alpha} \opt - \beta$, where the last inequality follows from weak duality.
We call this the {\em global analysis} framework.
On the other hand, it also suffices to show that $P^{i+1} - P^i \ge \tfrac{1}{\alpha} (D^{i+1} - D^i)$ for all $i$; 
summing over all $i$ we have $P^n - P^0 \ge \tfrac{1}{\alpha} (D^n - D^0)$, or 
$P^n \ge \tfrac{1}{\alpha} D^n - \beta$ for $\beta = \tfrac{1}{\alpha} D^0 - P^0$.
We call this the {\em local analysis} framework.
We will use both frameworks in this paper.

\paragraph{Remarks on Notation}

When $\Delta y = 1$ and $\SSS = \{0, 1\}^m$, our setting becomes \citet{BGMS11}'s setting, except for the representation of production cost functions. In \citet{BGMS11}, there is an increasing marginal production cost function $c_j$ for each integral unit of item $j$, whereas our production cost functions $f_j$ are the cumulative version. If $f_j(y_j)$'s further take on value $0$ for $0 \le y_j \le k$, and $+\infty$ otherwise, our setting essentially becomes multi-unit combinatorial auctions with multi-minded buyers, as considered by \citet{BGN03} and, more recently, \citet{BG13}.

\section{Fractional Case}\label{sec:frac}

In this section, we consider the fractional case of our problem, that is, we assume $\Delta y$ is infinitesimally small and let there be infinitely many buyers each of which buys infinitesimal units of items. The fractional case allows us to focus on the online nature of the problem, while ignoring subtle treatments needed to handle the integrality gap of the convex programs.

For the fractional case, we are able to characterize the optimal competitive ratio achievable by online mechanisms/algorithms.
We show that finding a pricing function $p : \mathbb{R}^+ \mapsto \mathbb{R}^+$ that is a feasible solution to the following differential equation for some constant $\beta$ is a sufficient and necessary condition of the existence of $\alpha$-competitive online mechanisms/algorithms:
\begin{equation}
\label{eq:diffeq3}
\textstyle 
\int_0^y p(\bar{y}) d\bar{y} - f(y) \geq \frac{1}{\alpha} \cdot f^*(p(y)) - \beta, \textrm{ for all $y\geq 0$}.
\end{equation}


\begin{theorem}
\label{thm:fracub}
If a monotonically increasing pricing function $p$ satisfies \pref{eq:diffeq3} for some constant $\beta$, then the corresponding pricing mechanism $\MMM_p$ is $\alpha$-competitive and incentive compatible.
\end{theorem}

\begin{theorem}
\label{thm:fraclb}
If there is an $\alpha$-competitive algorithm, then there exists a monotonically increasing pricing function $p$ that satisfies \pref{eq:diffeq3} for some constant $\beta$.
\end{theorem}

We note that the integral version of \pref{eq:diffeq3} is equivalent to the Structural Lemma by \citet{BGMS11} and, perhaps not surprisingly, our mechanisms are posted pricing mechanisms. 

Further, we define 
\[
 \alpha(f) = \inf \big\{ \alpha : \textrm{there exist constant $\beta$ and monotonically increasing $p$ so that \pref{eq:diffeq3} holds} \big\}.
\]


\begin{corollary}
For any $\epsilon > 0$, there is an $(\alpha(f) + \epsilon)$-competitive and incentive compatible mechanism.
\end{corollary}

\begin{corollary}
For any $\epsilon > 0$, there are no $(\alpha(f) - \epsilon)$-competitive algorithms.
\end{corollary}

We stress that our upper bound holds for incentive compatible mechanisms, while our lower bound holds for arbitrary algorithms. In this sense, the incentive compatibility constraint does not impose any additional difficulties for online combinatorial auctions with production costs. We note that the integral analogue of \pref{thm:fracub} is equivalent to the Structural Lemma in \citet{BGMS11}.

We present the proofs of \pref{thm:fracub} and \pref{thm:fraclb} in \pref{sec:fracub} and \pref{sec:fraclb}, respectively, and characterize $\alpha(f)$ for some specific production cost functions in \pref{sec:fraccase}.

\subsection{Mechanism (Proof of \pref{thm:fracub})}
\label{sec:fracub}



For notational simplicity, let $S_i$ be the utility-maximizing bundle that buyer $i$ purchases and $v_i$ be the value $v_{iS_i}$. 
By the definition of pricing mechanisms, we have $u_i = v_i - \sum_{j\in S_i} p_j^{i-1} \cdot \Delta y$. 
(Recall that we assume $a_{jS} \in \{0, \Delta y\}$.)

By the definition of the convex programs, we have feasible solutions and $P^n = \sum_i v_i - \sum_j f(y_j)$ and $D^n = \sum_i u_i + \sum_j f^*(p_j)$ where $y_j$ and $p_j$ are the final demand and price for item $j$. 
To lower bound $P^n$ by the $\tfrac{1}{\alpha}$ fraction of $D^n$, we first rewrite $\sum_i v_i$ as a function of variables $y_j$, $u_i$, and $p_j$:

\begin{lemma}
$\sum_i v_i = \sum_i u_i + \sum_j \int_0^{y_j} p(y) dy$.
\end{lemma}

\begin{proof}
Let $\chi_j(i,y)$ be an indicator function that is equal to $1$ if buyer $i$ buys from the $y$-th to the $(y + \Delta y)$-th units of item $j$ via bundle $S_i$. 
The buyer $i$'s utility is $u_i = v_i - \sum_j \sum_{y \in \Delta y \cdot \mathbb{N}} \chi_j(i,y) \Delta y \cdot p(y)$, where $y$ ranges over all the nonnegative integer multiples of $\Delta y$. Then, 
\[ \textstyle \sum_i v_i = \sum_i u_i + \sum_i \sum_j \sum_y \chi_j(i,y) \cdot \Delta y \cdot p(y).  \]
Next, we change the order of summation and account for the second term, i.e., the total payment, in a different way:
\begin{align*}
\textstyle \sum_{i,j,y} \chi_j(i,y) \cdot \Delta y \cdot p(y) &= \textstyle \sum_{j,y} \Delta y \cdot p(y) \sum_i \chi_j(i,y) \\
	& = \textstyle \sum_{j, y \leq y_j} \Delta y \cdot p(y).
\end{align*}
As we assume $\Delta y$ to be infinitesimally small in the fractional case, the above reduces to $\sum_j \int_0^{y_j} p(y) dy$. The lemma follows.
\end{proof}


Given the lemma, the primal objective is $P^n = \sum_i u_i + \sum_j \int_0^{y_j} p(y) dy - \sum_j f(y_j)$. 
We now show that $P^n \geq \frac{1}{\alpha} \cdot D^n - m\beta$, then the competitive ratio follows from the weak duality property.
This is equivalent to 
\[
\textstyle \sum_i u_i + \sum_j \int_0^{y_j} p(y)dy - \sum_j f(y_j) \geq \frac{1}{\alpha} \big( \sum_i u_i + \sum_j f^*(p_j) \big) - m\beta.
\]

Note that $\alpha \ge 1$ and $\sum_i u_i \ge 0$. Having positive $\sum_i u_i$ only helps the inequality. So, it suffices to show $\sum_j \int_0^{y_j} p(y) dy - \sum_j f(y_j) \geq \frac{1}{\alpha} \sum_j f^*(p_j) - m\beta$, which follows by summing up \pref{eq:diffeq3} over all items.

\subsection{Lower Bound (Proof of \pref{thm:fraclb})}
\label{sec:fraclb}

We consider a family of single-item instances parameterized by $v^* \ge 0$, $\{I_{v^*}\}_{v^* \geq 0}$, and show that if there is an online algorithm that is $\alpha$-competitive for all instances in the family, then we can construct a monotonically increasing feasible solution to the differential equation \pref{eq:diffeq3}.

The instance $I_{v^*}$ is defined as follows:
let there be a continuum of stages parameterized by $v$ starting from stage $0$ to stage $v^*$;
at stage $v$, let there be a continuum of buyers with value $v$ per unit of the item and a total demand of ${f^*}'(v)$.
Since ${f^*}'$ and $f'$ are inverses,  ${f^*}'(v)$ is the maximal amount of the item that can be produced at a marginal production cost of at most $v$ per unit.

Consider any online algorithm.
Let $Y(v, v^*)$ be a random variable denoting the amount of the item sold up to stage $v \le v^*$ in instance $I_{v^*}$.
Let $y(v, v^*)$ be the expected value of $Y(v, v^*)$ over the randomness of the algorithm.
Note that when the algorithm decides the allocation for buyers at stage $v$, it does not have any information about $v^*$ other than that $v^* \ge v$. 
Hence, the distribution of random variable $Y(v, v^*)$ is independent of the value of $v^*$ for any $v^* \ge v$.
We will omit the second argument and simply write $Y(v)$ and $y(v)$ in the rest of the proof.

We first show that if there is a competitive algorithm, then there is a feasible solution to the ``inverse'' of the differential equation \pref{eq:diffeq3}:

\begin{lemma}
\label{lem:fraclbinv}
If there is an $\alpha$-competitive algorithm for all instances $I_{v^*}$, $v^* \ge 0$, then there is a constant $\beta$ and a function $y(v)$ such that:
\begin{equation}
\label{eq:fracinverse}
\textstyle \int_0^{v^*} v d y(v) - f(y(v^*)) \ge \tfrac{1}{\alpha} f^*(v^*) - \beta, \textrm{ for } v^* \ge 0.
\end{equation}
\end{lemma}

\begin{proof}
In instance $I_{v^*}$, the optimal offline solution allocates ${f^*}'(v^*)$ units of the item to the buyers in the last stage and none to the buyers in previous stages. The optimal objective value of social welfare less the total production cost is
\[ \opt(v^*) = v^* \cdot {f^*}'(v^*) - f({f^*}'(v^*)) = f^*(v^*), \]
where the last equality follows from the definition of the convex conjugate function $f^*$ and properties of the complementary pair $v^*$ and ${f^*}'(v^*)$.

On the other hand, the objective value achieved by the algorithm is $\int_0^{v^*} v d Y(v) - f(Y(v^*))$.
By the linearity of the first term and convexity of $f$, the expected objective value of the algorithm is at least:
\[ \Ex[\textstyle \alg(v^*)] \leq \int_0^{v^*} v d y(v) - f(y(v^*)). \]

Therefore, if the algorithm is $\alpha$-competitive, then there exists $\beta$ such that for any $v^* \ge 0$, $\Ex[ \alg(v^*)] \ge \tfrac{1}{\alpha} \opt(v^*) - \beta$. 
Subsequently, there exists $y(v)$ that is a feasible solution to \pref{eq:fracinverse}.
\end{proof}

The differential equation \pref{eq:fracinverse} is essentially the same as \pref{eq:diffeq3}, except that a solution to \pref{eq:diffeq3} is $p(y)$ as a function of $y$ while a solution to \pref{eq:fracinverse} is $y(p)$ as a function of $p$.
In particular, if there is a feasible solution $y(p)$ to \pref{eq:fracinverse} that is strictly monotone, then its inverse function would be a feasible solution to \pref{eq:diffeq3} that is monotonically increasing.
The rest of this subsection is devoted to constructing a strictly monotone feasible solution.

Suppose the differential equation \pref{eq:fracinverse} is feasible for some $\alpha$ and $\beta$. For the same $\alpha$ and $\beta$ values, we let 
\[ \textstyle \underline{y}(v) = \inf \big\{ y(v) : \textrm{$y$ is feasible for \pref{eq:fracinverse}} \big\}. \]

\begin{lemma}
\label{lem:fraclbinv2}
\pref{eq:fracinverse} holds with equality for $\underline{y}$.
\end{lemma}

\begin{proof}
For a given $v^* \ge 0$, by the definition of $\underline{y}(v^*)$, there exists a feasible solution $\tilde{y}$ to \pref{eq:fracinverse} that takes value at most $\underline{y}(v^*) + \epsilon$ at $v^*$ for an arbitrarily small $\epsilon > 0$. In particular,
\begin{equation}
\label{eq:fracinvproof1}
\textstyle \int_0^{v^*} v d \tilde{y}(v) - f(\tilde{y}(v^*)) \ge \tfrac{1}{\alpha} f^*(v^*) - \beta.
\end{equation}
Note that $\int_0^{v^*} v d \tilde{y}(v) = v^* \tilde{y}(v^*) - \int_0^{v^*} \tilde{y}(v) dv$. 
By the definition of $\tilde{y}(v^*)$, the first term is at most $v^* \underline{y}(v^*) + v^* \epsilon$. 
Further, by the definition of $\underline{y}$, we have $\tilde{y}(v) \ge \underline{y}(v)$ for all $v$.
So, we have
\begin{align*}
\textstyle \int_0^{v^*} v d \tilde{y}(v) & \textstyle \le v^* \underline{y}(v^*) - \int_0^{v^*} \underline{y}(v) dv + v^* \epsilon \\
	&\textstyle = \int_0^{v^*} v d \underline{y}(v) + v^* \epsilon.
\end{align*}
Putting the above back to \pref{eq:fracinvproof1} and using the fact that $f(\underline{y}(v^*)) \leq f(\tilde{y}(v^*))$, we have 
\[
\textstyle \int_0^{v^*} v d \underline{y}(v) + v^* \epsilon - f(\underline{y}(v^*)) \ge \tfrac{1}{\alpha} f^*(v^*) - \beta \enspace.
\]
As it holds for arbitrarily small $\epsilon > 0$, it also holds for $\epsilon = 0$ in the limit. 
It follows that $\underline{y}$ is a feasible solution to the differential equation \pref{eq:fracinverse}. 
Further, \pref{eq:fracinverse} must hold with equality because otherwise we could further lower the value of $\underline{y}$ while maintaining feasibility, contradicting our choice of $\underline{y}$.  
\end{proof}

\begin{lemma}
$\underline{y}(v)$ is strictly increasing in $v$.
\end{lemma}

\begin{proof}
Suppose for contradiction that $\underline{y}(v)$ is not strictly increasing at $v = v^*$.
Consider \pref{eq:fracinverse} from $v^*$ to $v^* + dv$ for an infinitesimally small $dv$. 
Its left hand side remains the same while its right hand side strictly increases, contradicting the previous lemma that says \pref{eq:fracinverse} holds with equality for $\underline{y}$.
\end{proof}

\pref{thm:fraclb} then follows by that the inverse of $\underline{y}(v)$ is a monotonically increasing feasible solution to \pref{eq:diffeq3}.

\subsection{Case Study}
\label{sec:fraccase}

In this subsection, we use the differential equation \pref{eq:diffeq3} to study two specific families of production costs -- power production costs and concave marginal production costs.
It is easier to work with the following version of \pref{eq:diffeq3} without integrals: 
\begin{equation}\label{eq:diffeq}
\textstyle (p(y) - f'(y)) dy = \frac{1}{\alpha} {f^*}'(p(y)) dp, \textrm{ for all $y \geq 0$}.
\end{equation}

On the one hand, if $p(y)$ is a monotonically increasing feasible solution to \pref{eq:diffeq}, then integrating both sides we get that $p(y)$ is also a monotonically increasing feasible solution to the differential equation \pref{eq:diffeq3}. 
Note that it suffices to satisfy \pref{eq:diffeq} with inequality, namely, its left hand side greater than or equal to its right hand side.
For the purpose of finding a feasible solution, however, the more restricted equality version is more instructive.

On the other hand, if there is an $\alpha$-competitive algorithm/mechanism, then, as in \pref{sec:fraclb}, we can construct a feasible solution to \pref{eq:diffeq3} with equality. Thus, differentiating both sides we get \pref{eq:diffeq}.

Due to space constraints, we demonstrate only the theorems and defer the proofs to \pref{app:frac}. 

\begin{theorem}[Power Prod.\ Costs]
\label{thm:fracpoly}
If $f(y) = a y^{\gamma+1}$ is a power function with $\gamma \geq 1$, then $\alpha(f) = (\gamma+1)^{(\gamma+1)/\gamma}$.
In particular, the pricing mechanism $\MMM_p$ with $p(y) = (\gamma + 1) y^\gamma$ is $(\gamma+1)^{(\gamma+1)/\gamma}$-competitive.\footnote{Note that asymptotically $(\gamma+1)^{(\gamma+1)/\gamma} \approx e \gamma$ as $\gamma$ goes to infinity.}
\end{theorem}

\begin{theorem}[Concave Marginal Prod.\ Costs] \label{thm:fconcave}
If a cost function $f$ is such that $f'$ is differentiable, concave and strictly increasing, then $\alpha(f) \le 4$.
In particular, the pricing mechanism $\MMM_p$ with $p(y) = f'(2y)$ is $4$-competitive.
\end{theorem}

\noindent
We remark that our upper bound for concave marginal cost functions is tight, because $f(y) = y^2$ is a special case of concave marginal cost functions for which $\alpha(f) = 4$, by \pref{thm:fracpoly}.

Finally, we present a theorem that unifies above two cases.
Define $\Gamma^\times_{f,\lambda} = \max \big\{ 1, \max_{y>0} \tfrac{(\lambda-1)y f''(\lambda y)}{f'(\lambda y) - f'(y)}\big\}$. It represents how fast the value of $f''$ could increase when its argument is scaled by a factor of $\lambda$, because $\tfrac{f'(\lambda y) - f'(y)}{(\lambda-1)y} \approx f''(y)$. When $f'$ is concave, $\Gamma^\times_{f,\lambda} = 1$. When $f'(y) = y^\gamma$, $\Gamma^\times_{f,\lambda} = \tfrac{(\lambda^\gamma - \lambda^{\gamma-1}) \gamma}{\lambda^\gamma - 1}$.

\begin{theorem}[A Unified Theorem] \label{thm:fconvex}
For cost functions $f$ with $f'$ differentiable and strictly increasing, the pricing mechanism $\MMM_p$ with $p(y) = f'(\lambda y)$ is $\tfrac{\lambda^2}{\lambda-1} \Gamma_{f,\lambda}^\times$-competitive for any $\lambda>1$.
\end{theorem}


\section{Integral Case}
\label{sec:int}

In this section, we consider the integral case of online combinatorial auctions where the constant $\Delta y$ is 1 and the set of bundles is $\SSS = \{0, 1\}^m$.
We derive constant competitive mechanisms for broad classes of production cost functions $f$. 
Our competitive ratios are arbitrarily close to the fractional counterparts in \pref{sec:frac}, and are strictly better than those obtained by \citet{BGMS11}.
Our mechanisms are posted pricing schemes similar to their twice-the-index pricing scheme.

We use superscript $i$ to denote the current values of the primal and dual variables after serving buyer $i$ and before serving buyer $i+1$. For example, suppose $y^{i-1}_j$ units of each item $j$ have been allocated so far, and buyer $i$ observes item prices $p_j^{i-1}$ and chooses his utility-maximizing bundle $S_i$. Then, we update the primal and dual variables as follows: $x_{iS_i}=1$; $y^i_j = y^{i-1}_j + 1$ for all item $j$ in $S_i$; $u_i = v_{iS_i} - \sum_{j\in S_i} p^{i-1}_j$; $p^i_j = p(y^i_j)$ for all item $j$ according to the pricing function $p$. 

We use the following integral analogue of \pref{eq:diffeq}:

\begin{lemma}\label{lem:diffeq2}
If a monotonically increasing pricing function $p$ satisfies, for all $i$,
\begin{equation}\label{eq:diffeq2}
\textstyle p_j^{i-1} - (f(y_j^{i-1}+1) - f(y_j^{i-1})) \geq \frac{1}{\alpha} (f^*(p_j^i) - f^*(p_j^{i-1})), 
\end{equation}
then the pricing mechanism $\MMM_p$ is $\alpha$-competitive (and incentive compatible).
\end{lemma}

The main idea is that the integral case approximately reduces to the fractional case when the seller has sold sufficiently many copies of an item for ``nice-behaving'' cost functions.
Concretely, if the cost function satisfies $f(y_j^i) - f(y_j^{i-1}) \approx f'(y_j^{i-1}) dy_j^{i-1}$ (note that $dy_j^{i-1} = 1$) and $f^*(p_j^i) - f^*(p_j^{i-1}) \approx {f^*}'(p_j^{i-1}) dp_j^{i-1}$, then \pref{eq:diffeq2} is essentially the same as the inequality version of \pref{eq:diffeq}. The contributions when $y_j$ is not sufficiently large can be accounted for by the additive cost $\beta$.\footnote{Note that additive costs are necessary as shown in Lemma A.1 in \citet{BGMS11}.}
Due to the space constraint, we defer proofs to \pref{app:iproofs}.

\begin{theorem}[Power Prod.\ Costs]\label{thm:ipower}
For a power cost function $f(y) = a y^{\gamma+1}$, the pricing mechanism $\MMM_p$ with $p(y) = a (\gamma + 1)^2 y^\gamma$ is $\alpha$-competitive with a sufficiently large additive cost, i.e., $W(\MMM_p) \geq \frac{1}{\alpha} \opt - \sum_j (\frac{1}{\alpha} f^*(f'(\frac{2}{\epsilon})) + f(\frac{1}{\epsilon}-1))$ for $\epsilon>0$, where $\alpha = (1+\epsilon)^\gamma (\gamma + 1)^{(\gamma+1)/\gamma}$.
\end{theorem}

\begin{theorem}[Concave Marginal Prod.\ Costs] \label{thm:iconcave}
For a cost function $f$ with $f'$ concave, the pricing mechanism $\MMM_p$ with $p(y) = f'(2(y+1))$ is $4(1+\epsilon)$-competitive with a sufficiently large additive cost, i.e., $W(\MMM_p) \geq \frac{1}{4(1+\epsilon)} \opt - \sum_j (\frac{1}{4(1+\epsilon)} f^*(f'(\frac{2}{\epsilon})) + f(\frac{1}{\epsilon}-1))$ for $\epsilon>0$.
\end{theorem}

To unify the above two theorems, recall that $\Gamma^\times_{f,\lambda} = \max \big\{1, \max_{y>0} \frac{(\lambda-1)y f''(\lambda y)}{f'(\lambda y) - f'(y)} \big\}$ from \pref{sec:frac}; further, let $\Gamma_{f,\lambda,\tau}^+ = \max \big\{1, \max_{y\geq \tau} \frac{f''(y+\lambda)}{f''(y)} \big\}$. 
If $f'$ is concave, $\Gamma_{f,\lambda}^\times = \Gamma_{f,\lambda,\tau}^+ = 1$.
If $f'(y) = y^\gamma$, $\Gamma^\times_{f,\lambda} = \tfrac{(\lambda^\gamma - \lambda^{\gamma-1}) \gamma}{\lambda^\gamma - 1}$ and $\Gamma^+_{f,\lambda,\tau}$ is 1 for $\gamma \le 1$ and is $(1 + \tfrac{\lambda}{\tau})^{\gamma-1}$ otherwise. 

\begin{theorem}[A Unified Theorem] \label{thm:iconvex}
For cost functions $f$ with $f'$ differentiable and strictly increasing, the pricing mechanism $\MMM_p$ with $p(y) = f'(\lambda(y+1))$ is $\alpha$-competitive for $\alpha = \frac{(1+\epsilon) \lambda^2}{\lambda-1} \cdot \Gamma_{f,\lambda,\lambda/\epsilon}^+ \cdot \Gamma_{f,\lambda}^\times$ and arbitrarily small $\epsilon >0$.
\end{theorem}

\paragraph{Comparisons with Previous Results}

\citet{BGMS11} considered a nearly identical problem and showed constant competitive posted pricing mechanisms for several marginal production cost functions: linear, polynomial, and logarithmic functions. In this section, we designed constant competitive algorithms for broader classes of production cost functions. 
We show that our results apply in \citet{BGMS11}'s setting and improve those competitive ratios previously obtained. The two settings differ in the representation of production cost functions: we use (cumulative) production cost functions $f$ whereas they use marginal production cost functions $c$. 
For each class of linear, polynomial, and logarithmic marginal production cost functions, we construct a strictly convex and differentiable production cost function $f$ that matches $c$ on each integral unit, i.e., $f(y) = \sum_{l=1}^y c(l)$ for all (sufficiently large) integer values $y$.\footnote{The equality is necessary, since the cost contributes negatively to the social welfare and optimizing the modified social welfare objective with a multiplicative approximate cost function does not lead to an approximately optimal solution with respect to the original social welfare objective.}

For linear $c(y) = ay + b$, where $a,b \geq 0$, we use the production cost function $f(y) = \frac{a}{2}y^2 + (b+ \frac{a}{2})y$. By \pref{thm:iconcave}, we get a $4(1+\epsilon)$-competitive pricing mechanism, improving the previous ratio of $6$. 

For polynomial $c(y) = a y^d$, where $d>1$, we use production cost function $f(y) = a S_d(y)$ where $S_d(y)$ is given by Faulhaber's formula for the power sum $\sum_{\ell=1}^n \ell^d$. For sufficiently large $y$ values, we upper bound $\Gamma_{f,\lambda}^\times$ and $\Gamma_{f, \lambda, \tau}^+$ and use these bounds to get the competitive ratio of $(1+\epsilon) (1+\frac{1}{d})^{d+1} d$, improving the previous $(1+\epsilon) 4d$ ratio.

For logarithmic $c(y) = \ln(1+y)$, we use a production cost function $f$ with a continuous, piecewise-linear, and concave first derivative $f'$ such that $\int^{i+1}_i f'(y) dy = c(i+1)$ for all $i$. By \pref{thm:iconcave}, we get a competitive ratio of $4(1+\epsilon)$ which is strictly better than the previous competitive ratio of $2/\ln(3/2) \approx 4.93$. Our results further applies to any concave marginal cost functions.

Note that $\epsilon$ can be arbitrarily small subject to a tradeoff in the additive cost in all of the above results. 
See \pref{tab:comp} for a summary of our improvements and \pref{app:comp} for the details.

\section{Limited $k$-Supply Case}\label{sec:limsupply}

In this section, we consider the limited supply version of the online combinatorial auction problem, also known as the multi-unit combinatorial auction with multi-minded buyers.
We show how to apply our primal-dual approach to this setting and get competitive ratios matching those in \citet{BGN03} and \citet{BG13}. 
All proofs are in \pref{app:limsupply}. 

In the limited supply case, items are allocated integrally and there are exactly $k$ units of each item for sale.
In this setting, $\Delta y = 1$; $\SSS = \{0, 1\}^m$; and $f(y) = 0$ for $y \in [0, k]$ and $f(y) = +\infty$ for $y > k$.
\citet{BGN03} and \citet{BG13} showed an $O(k((m\rho)^{1/k} - 1))$-competitive algorithm, where $\rho = v_{\max}/ v_{\min}$, $v_{\max} = \max_{i,S} v_{iS}$, and $v_{\min} = \min_{i,S: v_{iS}>0} v_{iS}$.\footnote{Only the ratio $v_\textrm{max}/v_\textrm{min}$ matters. So we may assume $v_\textrm{min}$ is $1$ after scaling. Note the knowledge of $v_\textrm{max}$ is necessary to obtain a non-trivial competitive ratio, as shown in \cite{BG13}.} 
This competitive ratio is $O(\log (m\rho))$ when the supply is at least $k = \Omega(\log m)$.

We first briefly discuss the fractional case to build our intuition.
Again, it suffices to construct a pricing function that is a feasible solution to \pref{eq:diffeq} (where $f^*(p) = k p$). 
Further, since the values are bounded by $v_{\min}$ and $v_{\max}$, we may let $p(0) = v_{\min} / m$ without loss of generality (so that the initial price of all bundles are at most $v_{\min}$), and it suffices to satisfy \pref{eq:diffeq} for $v \ge 0$ s.t.\ $p(v) \le v_{\max}$.
Thus, \pref{eq:diffeq} becomes: for all $y \geq 0$ such that $p(y) \le v_{\max}$,
\[ \textstyle (p(y) - f'(y)) dy = \frac{k}{\alpha} dp.\]
Since $f'(y) = 0$ for $0 \le y \le k$ and $+\infty$ for $y > k$, we get that $p \cdot dy = \tfrac{k}{\alpha} dp$ for $0 \le y \le k$ and $p(k) \ge v_{\max}$.
By the first equation, $p(y) = p(0) \cdot \exp( \tfrac{\alpha}{k} y)$.
By the boundary condition that $p(0) = v_{\min} / m$ and $p(k) \ge v_{\max}$, let $\alpha = \ln(m \rho)$.
Thus, the fractional solution $p(y) = p(0) (m \rho)^{y/k}$ is a natural candidate solution.
For technical reasons, we decrease the starting price $p(0)$ by a factor of $2$ and use $p(y) = p(0) (2m \rho)^{y/k}$ in the integral case.


\begin{theorem}\label{thm:limsupply}
The pricing mechanism $\MMM_p$ with $p(y) = p(0) r^y$, where $p(0)=\frac{v_{\textrm{min}}}{2m}$ and $r=(2m \rho)^{1/k}$, is $\Theta(k((2m \rho)^{1/k} -1))$-competitive and incentive compatible for combinatorial auctions with supply $k$.
\end{theorem}

Recall that when $k = \Omega(\log m)$, the competitive ratio is $\Theta(\log (m \rho))$, since $k((2m \rho)^{1/k} -1) = k(e^{\frac{1}{k} \ln(2m \rho)} - 1) \approx k \frac{1}{k} \ln(2m \rho) =\ln(2m \rho)$. 
\citet{BG13} used a similar primal dual approach based on the standard linear program relaxations that impose the supply constraint as linear constraints (See \pref{app:limsupply}). 
Our approach is different in that we consider the supply constraint in a broader production cost model using convex programs, i.e., ``lifting'' the supply constraint into the objective. 
As a result, the pricing mechanism follows straightforwardly as a solution of a differential equation. 
We believe this approach will find further applications on similar problems.

Finally, we use our framework to show an almost matching logarithmic lower bound that applies to randomized algorithms, thus resolving an open problem by \citet{BG13} on if randomized algorithms can overcome the logarithmic lower bound.

\begin{theorem}
\label{thm:limsuplylb}
No online algorithms are $o(\tfrac{\log m}{\log\log m} + \log \rho)$-competitive for combinatorial auctions.
\end{theorem}



\bibliographystyle{chicago}
\bibliography{refs}

\begin{thebibliography}{}

\bibitem[\protect\citeauthoryear{Anand, Garg, and Kumar}{Anand
  et~al.}{2012}]{AGK12}
Anand, S., N.~Garg, and A.~Kumar (2012).
\newblock Resource augmentation for weighted flow-time explained by dual
  fitting.
\newblock In {\em SODA}, pp.\  1228--1241.

\bibitem[\protect\citeauthoryear{Balcan, Blum, and Mansour}{Balcan
  et~al.}{2008}]{BBM08}
Balcan, M.-F., A.~Blum, and Y.~Mansour (2008).
\newblock Item pricing for revenue maximization.
\newblock In {\em EC}, pp.\  50--59.

\bibitem[\protect\citeauthoryear{Bartal, Gonen, and Nisan}{Bartal
  et~al.}{2003}]{BGN03}
Bartal, Y., R.~Gonen, and N.~Nisan (2003).
\newblock Incentive compatible multi unit combinatorial auctions.
\newblock In {\em TARK}, pp.\  72--87.

\bibitem[\protect\citeauthoryear{Blum, Gupta, Mansour, and Sharma}{Blum
  et~al.}{2011}]{BGMS11}
Blum, A., A.~Gupta, Y.~Mansour, and A.~Sharma (2011).
\newblock Welfare and profit maximization with production costs.
\newblock In {\em FOCS}, pp.\  77--86.

\bibitem[\protect\citeauthoryear{Blumrosen and Nisan}{Blumrosen and
  Nisan}{2007}]{BN07}
Blumrosen, L. and N.~Nisan (2007).
\newblock Combinatorial auctions.
\newblock {\em Algorithmic Game Theory\/}.

\bibitem[\protect\citeauthoryear{Buchbinder and Gonen}{Buchbinder and
  Gonen}{}]{BG13}
Buchbinder, N. and R.~Gonen.
\newblock Incentive compatible multi-unit combinatorial auctions: a primal dual
  approach.
\newblock {\em Algorithmica (to appear)\/}.

\bibitem[\protect\citeauthoryear{Buchbinder and Naor}{Buchbinder and
  Naor}{2009}]{BN09}
Buchbinder, N. and J.~Naor (2009).
\newblock The design of competitive online algorithms via a primal-dual
  approach.
\newblock {\em Foundations and Trends in Theoretical Computer Science\/}~{\em
  3\/}(2-3), 93--263.

\bibitem[\protect\citeauthoryear{Chakraborty, Huang, and Khanna}{Chakraborty
  et~al.}{2009}]{CHK09}
Chakraborty, T., Z.~Huang, and S.~Khanna (2009).
\newblock Dynamic and non-uniform pricing strategies for revenue maximization.
\newblock In {\em FOCS}, pp.\  495--504.

\bibitem[\protect\citeauthoryear{Chawla, Hartline, Malec, and Sivan}{Chawla
  et~al.}{2010}]{CHMS10}
Chawla, S., J.~D. Hartline, D.~L. Malec, and B.~Sivan (2010).
\newblock Multi-parameter mechanism design and sequential posted pricing.
\newblock In {\em STOC}, pp.\  311--320.

\bibitem[\protect\citeauthoryear{Devanur}{Devanur}{2010}]{Dev10}
Devanur, N.~R. (2010).
\newblock Fisher markets and convex programs.
\newblock Manuscript, available at {\tt
  http://resear\\ch.microsoft.com/en-us/um/people/nikdev/}.

\bibitem[\protect\citeauthoryear{Devanur and Huang}{Devanur and
  Huang}{2014}]{DH14}
Devanur, N.~R. and Z.~Huang (2014).
\newblock Primal dual gives almost optimal energy efficient online algorithms.
\newblock In {\em SODA}, pp.\  1123--1140.

\bibitem[\protect\citeauthoryear{Devanur and Jain}{Devanur and
  Jain}{2012}]{DJ12}
Devanur, N.~R. and K.~Jain (2012).
\newblock Online matching with concave returns.
\newblock In {\em STOC}, pp.\  137--144.

\bibitem[\protect\citeauthoryear{Dughmi, Roughgarden, and Yan}{Dughmi
  et~al.}{2011}]{DRY11}
Dughmi, S., T.~Roughgarden, and Q.~Yan (2011).
\newblock From convex optimization to randomized mechanisms: toward optimal
  combinatorial auctions.
\newblock In {\em STOC}, pp.\  149--158.

\bibitem[\protect\citeauthoryear{Dughmi and Vondr{\'a}k}{Dughmi and
  Vondr{\'a}k}{}]{DV14}
Dughmi, S. and J.~Vondr{\'a}k.
\newblock Limitations of randomized mechanisms for combinatorial auctions.
\newblock {\em Games and Economic Behavior (to appear)\/}.

\bibitem[\protect\citeauthoryear{Gupta, Krishnaswamy, and Pruhs}{Gupta
  et~al.}{2012}]{GKP12}
Gupta, A., R.~Krishnaswamy, and K.~Pruhs (2012).
\newblock Online primal-dual for non-linear optimization with applications to
  speed scaling.
\newblock In {\em WAOA}, pp.\  173--186.

\bibitem[\protect\citeauthoryear{Lavi and Swamy}{Lavi and Swamy}{2011}]{LS11}
Lavi, R. and C.~Swamy (2011).
\newblock Truthful and near-optimal mechanism design via linear programming.
\newblock {\em J. ACM\/}~{\em 58\/}(6), 25.

\bibitem[\protect\citeauthoryear{Thang}{Thang}{2013}]{Tha13}
Thang, N.~K. (2013).
\newblock Lagrangian duality in online scheduling with resource augmentation
  and speed scaling.
\newblock In {\em ESA}, pp.\  755--766.

\bibitem[\protect\citeauthoryear{Vondr{\'a}k}{Vondr{\'a}k}{2008}]{Von08}
Vondr{\'a}k, J. (2008).
\newblock Optimal approximation for the submodular welfare problem in the value
  oracle model.
\newblock In {\em STOC}, pp.\  67--74.

\end{thebibliography}









\appendix

\section{LP for the Single-Unit Combinatorial Auction Problem}\label{app:singleunit}
In the single-unit combinatorial auction problem, the seller has $m$ items, represented by $[m]$, to allocate to $n$ buyers, represented by $[n]$. Let $\SSS$ be the the collection of all subsets of $[m]$. There are no production costs, but there is exactly 1 unit available for sale for each item. The seller's objective is to maximize the total social welfare which is exactly the total aggregate value of the buyers' respective allocated bundles. The following are the standard primal and dual linear program relaxations for this problem:
\begin{align*}
\textstyle \max_{x} ~~ & \textstyle \sum_i \sum_S v_{iS} x_{iS} \\
\forall i: \quad & \textstyle \sum_S x_{iS} \leq 1 \\
\forall j: \quad & \textstyle \sum_i \sum_{S: j\in S} x_{iS} \leq 1 \\
& x \geq 0 \\
\\
\textstyle \min_{u, p} ~~ & \textstyle \sum_i u_i + \sum_j p_j \\
\forall i, S: \quad & \textstyle u_i + \sum_{j\in S} p_j \geq v_{iS} \\
& u, p \ge 0 
\end{align*}
The primal variables $x_{iS}$ indicate whether or not buyer $i$ purchases bundle $S$. We have linear constraints that impose the conditions that a buyer purchases at most 1 bundle and that each item gets purchased at most once. We have dual variables $u$ and $p$: one variable $u_i$ for each buyer and one variable $p_j$ for each item. To minimize the dual objective function, we let $u_i = \max_S\{v_{iS} - \sum_{j\in S} p_j\}$ for all $i$, given the valuations $v_{iS}$ and variables $p_j$. We interpret $p_j$ as the price of item $j$ and $u_i$ as the utility of buyer $i$ for buying his utility-maximizing bundle $S$ at price $\sum_{j\in S} p_j$.

\section{Derivation}\label{app:deriv}
We derive the dual convex program ($D$) from the primal convex program ($P$) in Section~\ref{sec:prelim} using Lagrangian duality. For the linear constraints in $P$, we define the dual variables: $u_i \geq 0$ for each buyer $i$; $p_j \geq 0$ for each item $j$; $\mu_{iS} \geq 0$ for each buyer $i$ and bundle $S$; $\eta_j \geq 0$ for each item $j$. We define the Lagrangian function $L$, omitting the primal and dual variables in the function argument, as follows:
\begin{align*}
L = &\textstyle \sum_{i, S} v_{iS} x_{iS} - \sum_j f(y_j) + \sum_i \lambda_i (1- \sum_S x_{iS}) \\
	&\textstyle ~ + \sum_j \lambda'_j(y_j-\sum_{i,S:j\in S} a_{jS} x_{iS}) + \sum_{i,S} \alpha_{iS} x_{iS} \\
	&\textstyle ~ + \sum_j \beta_j y_j\\
	= &\textstyle \sum_j (y_j(\beta_j + \lambda'_j) - f(y_j)) + \sum_i \lambda_i \\
	&\textstyle ~ + \sum_{i,S} x_{iS}(v_{iS} - \lambda_i - \sum_{j\in S} \lambda'_j a_{jS} + \alpha_{iS}).
\end{align*}
Then, the dual program is $\min_{u,p,\mu,\nu \geq 0} \max_{x,y} L$. Given the dual variables, $\max_{x,y} L = \sum_j f^*(p'_j) + \sum_i \lambda_i$ as long as $p'_j \geq p_j$ for all item $j$ and $v_{iS} - \sum_{j\in S} a_{jS} p_j \leq u_i$ for all buyer $i$ and bundle $S$; otherwise, $\max_{x,y} L$ is unbounded. Note that $f^*$ is the convex conjugate of $f$, i.e., $f^*(p) = \max_y \{py - f(y)\}$. When $f$ is strictly convex and differentiable, which is true for broad classes of production cost functions we consider in this paper, $f^*$ is increasing and we have $p'_j = p_j$ for all item $j$. Finally, we obtain the dual program ($D$).


\section{Missing Proofs from \pref{sec:frac}}
\label{app:frac}

\begin{proof}[\pref{thm:fracpoly}]
Note that $\alpha(f)$ is scale invariant: if $p(v)$ is a feasible solution for \pref{eq:diffeq} with respect to $f$, then $a \cdot p(v)$ is a feasible solution with respect to $a \cdot f$ for any constant $a$.
So, we may assume $a = \tfrac{1}{\gamma+1}$ without loss of generality. 
We have $f'(y) = y^\gamma$, $f^*(p) = \tfrac{\gamma}{\gamma+1} p^{(\gamma+1)/\gamma}$, and ${f^*}'(p) = p^{1/\gamma}$.
Then, \pref{eq:diffeq} becomes
\begin{equation}
\label{eq:fracpower}
(p(y) - y^\gamma) dy = \tfrac{1}{\alpha} p(y)^{1/\gamma} dp, \textrm{ for all $y \geq 0$.}
\end{equation}

\paragraph{Upper Bound:}
Note that $p(y) = (\lambda y)^\gamma$ is a natural candidate because all terms in the above differential equation would have the same degree in $y$ and, thus, the contribution of $y$ would cancel out.
Concretely, the differential equation becomes $(\lambda^\gamma - 1) = \tfrac{\gamma \lambda^{\gamma+1}}{\alpha}$. 
Choosing $\lambda = (\gamma+1)^{1/\gamma}$ to maximize $\tfrac{1}{\alpha} = \tfrac{1}{\gamma} \tfrac{1}{\lambda} (1 - \tfrac{1}{\lambda^\gamma})$, we have that the pricing mechanism $\MMM_p$ with $p(y) = (\gamma + 1) y^\gamma$ is $(\gamma+1)^{(\gamma+1)/\gamma}$-competitive.
So, $\alpha(f) \le (\gamma+1)^{(\gamma+1)/\gamma}$.

\paragraph{Lower Bound:}
Suppose $p(y)$ is a feasible solution to the differential equation \pref{eq:fracpower} for some constant $\alpha$. 
We need to show that $\alpha \ge (\gamma+1)^{(\gamma+1)/\gamma}$.

Let $\epsilon > 0$ be an arbitrarily small constant.
We let $c_0 = (1 - \epsilon) \tfrac{\gamma}{\alpha}$ and inductively define $c_{i+1} = (1 - \epsilon) \tfrac{\gamma}{\alpha (1-c_i^\gamma)}$ for $i \ge 0$.
We will first show the following lemma:

\begin{lemma}
For any $i \ge 0$, there exists $y_i$ such that for all $y > y_i$, $y > c_i p(y)^{1/\gamma}$.
\end{lemma}

\begin{proof}
We will prove the lemma by induction. 
Let us start with the base case.
Since $y \ge 0$, we have 
\[ p \cdot dy \ge \tfrac{1}{\alpha} p^{1/\gamma} dp. \] 
Rearranging terms, we have $dy \ge \tfrac{\gamma}{\alpha} d p^{1/\gamma} = \tfrac{c_0}{1 - \epsilon} d p^{1/\gamma}$ and, thus, 
$y \ge \tfrac{c_0}{1 - \epsilon} (p(y)^{1/\gamma} - p(0)^{1/\gamma})$.
Note that for \pref{eq:fracpower} to be feasible, we have $p(y) \ge y^\gamma$.
For all $y > \tfrac{1}{\epsilon} p(0)^{1/\gamma}$, we have $p(y) > \tfrac{1}{\epsilon^\gamma} p(0)$.
Putting together, we have $y > c_0 p(y)^{1/\gamma}$.
So the base case follows.

Suppose $y > c_i p(y)^{1/\gamma}$ for all $y > y_i$. 
Then, by \pref{eq:fracpower}, we have
\[ (1 - c_i^\gamma) p \cdot dy \ge \tfrac{1}{\alpha} p^{1/\gamma} dp. \] 
Rearranging terms, we have $dy \ge \tfrac{\gamma}{\alpha (1 - c_i^\gamma)} d p^{1/\gamma} = \frac{c_{i+1}}{1-\epsilon} d p^{1/\gamma}$ for all $y \ge y_i$ and, thus, $y - y_i \ge \frac{c_{i+1}}{1-\epsilon} (p(y)^{1/\gamma} - p(y_i)^{1/\gamma})$.
For all $y > \tfrac{1}{\epsilon} p(y_i)^{1/\gamma}$, we have $p(y) \ge y^\gamma > \tfrac{1}{\epsilon^\gamma} p(y_i)$.
Further note that $y_i \ge 0$.
Putting together, we have $y > c_{i+1} p(y)^{1/\gamma}$.
\end{proof}

Recall that for \pref{eq:fracpower} to be feasible, we must have $p(y) \ge y^\gamma$.
So $c_i < 1$ for all $i$ and $\{c_i\}_{i\geq 0}$ is an increasing sequence. 
Since $\{c_i\}_{i \ge 0}$ is increasing and bounded, it converges.
For sufficiently large $i$, we have $c_{i+1} \le (1 + \epsilon) c_i$.
So $(1 - \epsilon) \tfrac{\gamma}{\alpha (1-c_i^\gamma)} = c_{i+1} \le (1 + \epsilon) c_i$ and, thus, $\tfrac{1}{\alpha} \le \tfrac{1}{\gamma} \tfrac{1 + \epsilon}{1 - \epsilon} c (1 - c^\gamma) \le \tfrac{1 + \epsilon}{1 - \epsilon} (\gamma+1)^{-(\gamma+1)/\gamma}$.
As this holds for arbitrarily small $\epsilon$, we have $\alpha(f) \geq (\gamma+1)^{(\gamma+1)/\gamma}$.
\end{proof}

\begin{proof}[\pref{thm:fconcave}]
Note that for the upper bound, it suffices to satisfy \pref{eq:diffeq} with inequality, i.e., the left hand side is greater than or equal to the right hand side.
For $\lambda > 1$ to be determined later, let $p(y) = f'(\lambda y)$. 
Then, the inequality version of \pref{eq:diffeq} becomes
\[ f'(\lambda y) - f'(y) \ge \tfrac{1}{\alpha} \lambda^2 y f''(\lambda y) \textrm{, for all $y\geq 0$}. \]
Since $f'$ is concave, $f'(\lambda y) - f'(y) \geq (\lambda-1)y f''(\lambda y)$.
So it suffices to show $\lambda-1 \geq \frac{1}{\alpha} \lambda^2$. 
To minimize $\alpha$, we choose $\lambda = 2$ to maximize $\tfrac{\lambda-1}{\lambda^2}$. 
For $\lambda=2$ and, thus, $p(y) = f'(2y)$, we get $\alpha = 4$.
\end{proof}

\begin{proof}[\pref{thm:fconvex}]
By the definition of $\Gamma^\times_{f,\alpha}$, $\Gamma^\times_{f,\alpha} \geq \tfrac{(\lambda-1)yf''(\lambda y)}{f'(\lambda y) - f'(y)}$ or, equivalently, 
\[ f'(\lambda y) - f'(y) \geq \tfrac{\lambda-1}{\Gamma^\times_{f,\alpha}} y f''(\lambda y), \textrm{ for all $y\geq 0$.}\]

Hence, it is sufficient to show $\tfrac{\lambda-1}{\Gamma^\times_{f,\lambda}} y f''(\lambda y) \geq \tfrac{1}{\alpha} \lambda^2 y f''(\lambda y)$ or $\tfrac{\lambda-1}{\lambda^2 \Gamma^\times_{f,\alpha}} \geq \tfrac{1}{\alpha}$. The rest of the analysis is identical to those of \pref{thm:fracpoly} and \pref{thm:fconcave}. 
\end{proof}

\section{Missing Proofs and Details from \pref{sec:int}}

\subsection{Missing Proofs}\label{app:iproofs}
\begin{proof}[\pref{lem:diffeq2}]
We first show that $P^i - P^{i-1} \geq \frac{1}{\alpha} (D^i - D^{i-1})$ for all $i$. For notational simplicity, let $S_i$ be the utility-maximizing bundle that buyer $i$ purchases and $v_i$ be the value $v_{iS_i}$, such that $u_i = v_i - \sum_{j\in S_i} p^{i-1}_j$. Note 
\begin{align*}
\textstyle P^i - P^{i-1} &\textstyle = v_i - \sum_j (f(y_j^i) - f(y_j^{i-1})) \\
	&\textstyle = v_i - \sum_{j\in S_i} (f(y_j^{i-1}+1) - f(y_j^{i-1})),
\end{align*}
and 
\begin{align*}
\textstyle D^i - D^{i-1} &\textstyle = u_i + \sum_{j\in S_i} (f^*(p_j^i) - f^*(p_j^{i-1})) \\
	&\textstyle = v_i - \sum_{j\in S_i} p_j^{i-1} + \sum_{j\in S_i} (f^*(p_j^i) - f^*(p_j^{i-1})).
\end{align*}
Then, $P^i - P^{i-1} \geq \frac{1}{\alpha} (D^i - D^{i-1})$ is equivalent to
\[
\textstyle(1-\frac{1}{\alpha}) v_i - \sum_{j\in S_i}(f(y_j^{i-1}+1) - f(y_j^{i-1})) \geq -\frac{1}{\alpha} \sum_{j\in S_i} p_j^{i-1} + \frac{1}{\alpha}\sum_{j\in S_i} (f^*(p_j^i) - f^*(p_j^{i-1})).
\]

As buyer $i$ maximizes his utility, $v_i \geq \sum_{j\in S_i} p_j^{i-1}$ and the left hand side of the last inequality is at least $(1-\frac{1}{\alpha}) \sum_{j\in S_i} p_j^{i-1} - \sum_{j\in S_i}(f(y_j^{i-1}+1) - f(y_j^{i-1}))$. After some algebra, it is sufficient to show 
\[
\textstyle \sum_{j\in S_i} p_j^{i-1} - \sum_{j\in S_i} (f(y_j^{i-1}+1) - f(y_j^{i-1})) \geq \frac{1}{\alpha} \sum_{j\in S_i} (f^*(p_j^i) - f^*(p_j^{i-1})),
\]
which follows from \pref{eq:diffeq2}. Therefore,  $P^i - P^{i-1} \geq \frac{1}{\alpha} (D^i - D^{i-1})$ for all $i=1, \ldots, n$.

Summing the inequality $P^i - P^{i-1} \geq \frac{1}{\alpha} (D^i - D^{i-1})$ over all $i$, we obtain $P^n - P^0 \geq \frac{1}{\alpha} (D^n - D^0)$. After rearranging terms and using the weak duality, we get 
\[\textstyle P^n \geq \frac{1}{\alpha} D^n - (\frac{1}{\alpha} D^0 - P^0) \geq \frac{1}{\alpha} \opt - (\frac{1}{\alpha} D^0 - P^0).\] 
As the final value of the primal objective function is the total social welfare that pricing mechanism $\MMM_p$ achieves, $W(\MMM_p)$, the lemma follows.
\end{proof}

\begin{proof}[\pref{thm:ipower}]
For $\lambda>1$ to be chosen later, we let the pricing function $p$ to be $p(y) = f'(\lambda \cdot (y+1))$ and initialize $y_j = \frac{1}{\epsilon}-1$ for each item $j$ in Step 1 of $\MMM_p$. Other primal and dual variables are still initialized to 0, except for the prices $p_j$ which depend on variables $y_j$. The left hand side of \pref{eq:diffeq2} can be lower bounded as follows:
\begin{align*}
\textstyle \operatorname{LHS} &= \textstyle f'(\lambda(y_j^{i-1}+1)) - (f(y_j^{i-1}+1) - f(y_j^{i-1})) \\
	&\textstyle \geq f'(\lambda(y_j^{i-1}+1)) - f'(y_j^{i-1}+1),
\end{align*}
where the inequality follows by the mean value theorem and the fact that $f'$ is an increasing function. Similarly, the right hand side of \pref{eq:diffeq2} can be upper bounded: 
\begin{align*}
\textstyle \operatorname{RHS} &\leq \textstyle \frac{1}{\alpha} (p_j^i - p_j^{i-1}) {f^*}'(p_j^i) \\
	&\textstyle \leq \frac{1}{\alpha} \lambda^2 (y_j^{i-1}+2) f''(\lambda(y_j^{i-1}+2)) \\
	&\textstyle \leq \frac{1}{\alpha} \lambda^2 (1+\epsilon)^\gamma (y^{i-1}_j+1) f''(\lambda(y^{i-1}_j+1)),
\end{align*}
where the first inequality follows from the mean value theorem and that ${f^*}'$ is an increasing function; the second from the convexity of $f'$ and that ${f^*}'$ and $f'$ are inverses; and the third from the initial conditions.

It suffices to choose $\lambda>1$ such that 
\[
\textstyle f'(\lambda(y_j^{i-1}+1)) - f'(y_j^{i-1}+1) \geq \frac{1}{\alpha} \lambda^2 (1+\epsilon)^\gamma (y^{i-1}_j+1) f''(\lambda(y^{i-1}_j+1)).
\]
For power cost functions, this reduces to $\frac{1}{(1+\epsilon)^\gamma} \frac{\lambda^\gamma-1}{\gamma \lambda^{\gamma+1}} \geq \frac{1}{\alpha}$. As in \pref{thm:fracpoly}, we choose $\lambda=(\gamma+1)^{1/\gamma}$ and obtain $\alpha = (1+\epsilon)^\gamma (\gamma+1)^{(\gamma+1)/\gamma}$. Note that $P_0 = - \sum_j f(\frac{1}{\epsilon}-1)$ and $D_0 = \sum_j f^*(f'(\frac{2}{\epsilon}))$. By \pref{lem:diffeq2}, the theorem statement follows.
\end{proof}

\begin{proof}[\pref{thm:iconcave}]
For $\lambda>1$ to be chosen later, we let the pricing function $p$ to be $p(y) = f'(\lambda \cdot (y+1))$ and initialize $y_j = \frac{1}{\epsilon}-1$ for each item $j$ in Step 1 of $\MMM_p$. As shown in the proof of \pref{thm:ipower}, the left hand side of \pref{eq:diffeq2} is lower bounded by $f'(\lambda(y_j^{i-1}+1)) - f'(y_j^{i-1}+1)$, and the right hand side is upper bounded as follows:
\begin{align*}
\textstyle \operatorname{RHS} &\textstyle \leq \frac{1}{\alpha}(f'(\lambda(y_j^{i-1}+2)) - f'(\lambda(y_j^{i-1}+1)))\lambda(y_j^{i-1}+2) \\
	& \textstyle \leq \frac{1}{\alpha} \lambda^2 (y_j^{i-1}+2) f''(\lambda(y_j^{i-1}+1)) \\
	& \textstyle \leq \frac{1}{\alpha} (1+\epsilon) \lambda^2 (y_j^{i-1}+1) f''(\lambda(y_j^{i-1}+1)),
\end{align*}
where the first inequality follows the same reasoning as in \pref{thm:ipower}; the second from the concavity of $f'$; and the third follows from the initial conditions $y_j^i\geq \frac{1}{\epsilon}-1$. 

Then, it is sufficient to choose $\lambda>1$ such that 
\[
\textstyle f'(\lambda(y_j^{i-1}+1)) - f'(y_j^{i-1}+1) \geq \frac{1}{\alpha} (1+\epsilon) \lambda^2 (y_j^{i-1}+1) f''(\lambda(y_j^{i-1}+1)).
\]
By the same reasoning as in \pref{thm:fconcave}, we choose $\lambda=2$ and obtain $\frac{1}{\alpha}=\frac{1}{4(1+\epsilon)}$. Note that $P_0 = - \sum_j f(\frac{1}{\epsilon}-1)$ and $D_0 = \sum_j f^*(f'(\frac{2}{\epsilon}))$. By \pref{lem:diffeq2}, the theorem statement follows.
\end{proof}

\begin{proof}[\pref{thm:iconvex}]
For $\lambda>1$ to be chosen later, we let the pricing function $p$ to be $p(y) = f'(\lambda \cdot (y+1))$ and initialize $y_j = \frac{1}{\epsilon}-1$ for each item $j$. The left hand side of \pref{eq:diffeq2} can be lower bounded as follows:
\begin{align*}
\textstyle \operatorname{LHS} & \textstyle \geq f'(\lambda(y_j^{i-1}+1)) - f'(y_j^{i-1}+1) \\
	& \textstyle \geq \frac{\lambda-1}{\Gamma_{f,\lambda}^\times} (y_j^{i-1} + 1) f''(\lambda(y_j^{i-1}+1)),
\end{align*}
where the first inequality follows the same reasoning as in \pref{thm:ipower} and the second from the definition of $\Gamma_{f,\lambda}^\times$. Similarly, the right hand side of \pref{eq:diffeq2} can be upper bounded: 
\begin{align*}
\textstyle \operatorname{RHS} &\textstyle \leq \frac{1}{\alpha}(f'(\lambda(y_j^{i-1}+2)) - f'(\lambda(y_j^{i-1}+1)))\lambda(y_j^{i-1}+2) \\
	& \textstyle \leq \frac{1}{\alpha} \lambda^2 \Gamma_{f, \lambda, \lambda/\epsilon}^+ (y_j^{i-1}+2) f''(\lambda(y_j^{i-1}+1)) \\
	& \textstyle \leq \frac{1}{\alpha} (1+\epsilon) \lambda^2 \Gamma_{f, \lambda, \lambda/\epsilon}^+ (y_j^{i-1}+1) f''(\lambda(y_j^{i-1}+1)),
\end{align*}
where the first inequality follows the same reasoning as in \pref{thm:ipower}; the second from the definition of $\Gamma_{f, \lambda, \lambda/\epsilon}^+$; and the third from the initial conditions $y_j^i\geq \frac{1}{\epsilon}-1$.

Then, it is sufficient to choose $\lambda$ such that 
\[
\textstyle \frac{\lambda-1}{\Gamma_{f,\lambda}^\times} (y_j^{i-1} + 1) f''(\lambda(y_j^{i-1}+1)) \geq \frac{1}{\alpha} (1+\epsilon) \lambda^2 \Gamma_{f, \lambda, \lambda/\epsilon}^+ (y_j^{i-1}+1) f''(\lambda(y_j^{i-1}+1)),
\]
which is equivalent to $\alpha \geq \frac{(1+\epsilon)\lambda^2}{\lambda-1} \cdot \Gamma_{f, \lambda, \lambda/\epsilon}^+ \cdot \Gamma_{f,\lambda}^\times$. From here, we follow the same reasoning as in \pref{thm:iconcave}. 
\end{proof}

\subsection{Detailed Comparisons with \citet{BGMS11}}\label{app:comp}
We construct a strictly convex and differentiable production cost functions $f$ given a marginal production cost function $c$ such that $f(y) = \sum_{l=1}^y c(l)$ for all sufficiently large integer values $y$ and apply results in \pref{sec:int} to get constant competitive algorithms. The equality is necessary, since the cost contributes negatively to the social welfare and optimizing the modified social welfare objective with a multiplicative approximate cost function does not lead to an approximately optimal solution with respect to the original social welfare objective.

\paragraph{Linear Marginal Costs} 

Assume that the marginal production cost function is a linear function, $c(y) = ay + b$, where $a,b \geq 0$. We use the production cost function $f(y) = \frac{a}{2}y^2 + (b+ \frac{a}{2})y$. It is straightforward to check that $f$ is strictly convex (over the domain $y \geq 0$) and differentiable and that, more importantly, $f(y) = \sum_{l=1}^y c(l)$ for all integer values $y$. Furthermore, note that $f'$ is concave. By \pref{thm:iconcave}, we obtain a $4(1+\epsilon)$-competitive pricing algorithm with some additive cost. This improves upon the previous competitive ratio of 6. 

\paragraph{Polynomial Marginal Costs}

Assume the marginal production cost function is a polynomial, $c(y) = a y^d$, where $d>1$ and $a>0$. Let $S_d(y)$ be the degree $d+1$ polynomial function given by Faulhaber's formula for the integer power sum $\sum_{l=1}^n l^d$, i.e., 
$$\textstyle S_d(n) = \frac{1}{d+1} \sum_{l=1}^{d+1} (-1)^{\delta_{ld}} \binom{d+1}{l} B_{d+1-l} n^l,$$ where $\delta_{ld}$ is the Kronecker delta and $B_l$ is the $l$-th Bernoulli number. We use the production cost function $f(y) = a\cdot S_d(y) = a_{d+1} y^{d+1} + a_d y^d + \cdots + a_0$. Note that $a_{d+1}$ and $a_d$ are positive. For sufficiently large values of $y$, $f''(y)>0$ and $f$ is strictly convex. As the following competitive analysis applies to $f$ when $y$ is sufficiently large, we may modify $f$ for small values of $y$ for the sake of strict convexity over the whole domain. For any $0 < \epsilon' < 1$ and sufficiently large values of $y$, we have the following inequalities:
\begin{align*}
&\textstyle f''(\lambda y) \leq (d+1) d a_{d+1} (1+\epsilon') (\lambda y)^{d-1}, \\
&\textstyle f''(y+\lambda) \leq (d+1) d a_{d+1} (1+\epsilon') (y+\lambda)^{d-1}, \textrm{ and}\\
&\textstyle f''(y) \geq (d+1) d a_{d+1} y^{d-1}.
\end{align*} 
For such large $y$ values, we can effectively upper bound $\Gamma_{f, \lambda}^\times$ and $\Gamma^+_{f, \lambda, \tau}$: 
\begin{align*}
&\textstyle \Gamma^\times_{f, \lambda} \leq (1+\epsilon')\lambda^{d-1}, \textrm{ and }  \\
&\textstyle \Gamma^+_{f, \lambda, \tau} \leq (1+\epsilon') (1+ \lambda/y)^{d-1} \leq (1+\epsilon')^d.
\end{align*}
We use these upper bounds in place of $\Gamma_{f, \lambda}^\times$ and $\Gamma^+_{f, \lambda, \tau}$ in \pref{thm:iconvex} for sufficiently small $\epsilon'$ to get the competitive ratio of $(1+\epsilon) (1+\frac{1}{d})^{d+1} d$ for any $\epsilon>0$. This improves upon the competitive ratio of $(1+\epsilon) 4d$.

\paragraph{Logarithmic Marginal Costs}

Assume the marginal production cost function is logarithmic, $c(y) = \ln(1+y)$. We construct a continuous, piecewise-linear, and concave $f'$ such that the corresponding production cost function $f$ satisfies $f(y) = \sum_{l=1}^y c(l)$ for all integer values $y$. Then, by \pref{thm:iconcave}, we get the competitive ratio of $4(1+\epsilon)$ which is strictly better than the previously obtained competitive ratio of $2 / \ln(3/2) \approx 4.93$. Our construction applies more generally for any concave marginal production cost function $c$ with a convex first derivative. 

Let $f(0)=0$. We define $f'$ at integer points and let $f'$ be linear over each interval $[i, i+1]$ as follows. In order to have $f(y) = \sum_{l=1}^y c(l)$ for integer values of $y$, we need to have $\int^{i+1}_i f'(y) dy = c(i+1)$ for all $i$. To this end, we let 
$$f'(i) = c(i+1) - \delta(i) \textrm{ and } f'(i+1) = c(i+1) + \delta(i),$$ 
for a suitable nonnegative function $\delta$ to be chosen later. Note $f'(i+1) = c(i+1) + \delta(i) = c(i+2) - \delta(i+1)$. It follows that 
$$\delta(i+1) = c(i+2) - c(i+1) - \delta(i), \textrm{for all $i$},$$ 
and any single value $\delta(i)$ completely determines the function $\delta$ and, consequently, $f'$. For the concavity of $f'$, we note the slope of $f'$ over $[i,i+1]$ is $2\delta(i)$ and it is sufficient to have $\delta(i)$ decreasing in $i$. Furthermore, since $\delta(i) \geq \frac{c(i+2) - c(i+1)}{2}$ implies $\delta(i+1) = c(i+2) - c(i+1) - \delta(i) \leq 2\delta(i) - \delta(i) = \delta(i)$, it is sufficient to have $\delta(i) \in J_i = [g(i+1), g(i)]$ for all $i$, where $g(i) = \frac{c(i+1) - c(i)}{2}$. 

We now show how to choose a value for $\delta(0)\in J_0$ such that $\delta(i) \in J_i$ for all $i$. We reduce the interval $[0, g(0)]$ to $[g(1), g(0)]$ by inductively mapping the points $g(i)$ to $\bar{g}(i)$ as follows: 
\begin{align*}
&\bar{g}(2j) = \bar{g}(2j-1) + (g(2j-1) - g(2j)), \textrm{ and }\\
&\bar{g}(2j+1) = \bar{g}(2j) - (g(2j) - g(2j+1)).
\end{align*}
Accordingly, we map the values $\delta(i)$ to $\bar{\delta}(i)$ in the interval $[g(1), g(0)]$. Note that $\bar{\delta}$ values superimpose onto a single point, say $\delta(0)$, in the interval. In addition, the sequence of intervals $\bar{J}_0, \bar{J}_1, \ldots$, where $\bar{J}_{2j} = [\bar{g}(2j+1), \bar{g}(2j)]$ and $\bar{J}_{2j+1} = [\bar{g}(2j+1), \bar{g}(2j+2)]$, are nested by the convexity of $c'$. It is sufficient to choose $\delta(0)$ such that it is included in $\bar{J}_i$ for all $i$. As there exists a point included in the intersection of any infinite sequence of nested intervals, we let $\delta(0)$ to be this value. It is straightforward to check that the resulting $f$ is strictly convex. The construction is complete.

\section{Additional Materials for \pref{sec:limsupply}}
\label{app:limsupply}

\subsection{LP for the Combinatorial Auction with Limited Supply}\label{app:limsupply}
In the combinatorial auction with limited supply, the seller has $m$ items, represented by $[m]$, to allocate to $n$ buyers, represented by $[n]$. Let $\SSS$ be the collection of all subsets of $[m]$. There are no production costs, but there are exactly $k$ units available for sale for each item. The seller's objective is to maximize the total social welfare which is exactly the total value of the buyers of their respective allocated bundles. The following are the standard primal and dual linear program relaxations for this problem, which was used in Buchbinder and Gonen~\cite{BG13}:

\begin{align*}
\textstyle \max_{x} ~~ & \textstyle \sum_i \sum_S v_{iS} x_{iS} \\
\forall i: \quad & \textstyle \sum_S x_{iS} \leq 1 \\
\forall j: \quad & \textstyle \sum_i \sum_{S: j\in S} x_{iS} \leq k \\
& x \geq 0 \\
\\
\textstyle \min_{u, p} ~~ & \textstyle \sum_i u_i + \sum_j p_j \\
\forall i, S: \quad & \textstyle u_i + \sum_{j\in S} p_j/k \geq v_{iS} \\
& u, p \geq 0 
\end{align*}

The primal and dual variables are the same as in the linear program relaxations for the single-unit combinatorial auction problem in Appendix~\ref{app:singleunit}.

\subsection{Missing Proof of \pref{thm:limsupply}}

\begin{proof}[\pref{thm:limsupply}]
We first show that at most $k$ units are produced and allocated for each item. Since $r = (2m \rho)^{1/k}$, the price for the $(k+1)$-th unit of any item is $p(k) = p(0) r^k = p(0) 2m \rho = v_{\max}$. By the assumption that buyers prefer the empty bundle to any nonempty bundle yielding the same utility of 0 and the definition of $v_{\max}$, no buyers will buy a bundle containing the $(k+1)$-th unit. Hence, at most $k$ units are sold for each item.

We now analyze the competitive ratio of the pricing mechanism $\MMM_p$. By the construction of the production cost function $f$ and the pricing function $p(y)$, \pref{eq:diffeq2} reduces to $p_j^{i-1} \geq \frac{k}{\alpha} (r-1) p_j^{i-1}$. The last inequality holds for $\alpha \geq k (r-1) = k((2m\rho)^{1/k}-1)$, so we let $\alpha = k((2m\rho)^{1/k}-1)$ and \pref{eq:diffeq2} holds for all $i$. By Lemma~\ref{lem:diffeq2}, 
$$\textstyle W(\MMM_p) \geq \frac{1}{\alpha} \opt - (\frac{1}{\alpha} D^0 - P^0).$$ 
Note that $P^0=0$ and $D^0 = \frac{1}{2} v_{\min} k \leq \frac{1}{2} \opt$. Furthermore, at least one unit of an item is bought by a buyer since $p(0)$ is small enough and there is a bundle of value at most $v_{\min}/2$, and  $P^n \not= P_0$ and $D^n \not= D_0$. Then, 
$$\textstyle W(\MMM_p) \geq \frac{1}{\alpha} \opt - \frac{1}{\alpha} D^0 \geq \frac{1}{2 \alpha} \opt.$$
Consequently, the final competitive ratio is $2 \alpha = \Theta(k((2m \rho)^{1/k} - 1))$.
\end{proof}

\subsection{Missing Proof of \pref{thm:limsuplylb}}

The proof follows the same framework as \pref{thm:fraclb}.
We present only a sketch in this paper.

\begin{proof}[\pref{thm:limsuplylb}]
We first prove that the competitive ratio is at least $\Omega(\log \rho)$ even for $m = 1$.
Consider a continuum of stages parameterized by $v^*$, from stage $v_{\min}$ to stage $v_{\max}$.
At each stage $v^* \geq v_{\min}$, let there be $k$ buyers with value $v^*$ for one copy of the item. 
Using the same argument as in \pref{lem:fraclbinv} and \pref{lem:fraclbinv2}, if there is an $\alpha$-competitive algorithm for all instances $v^* \ge 0$, then there is a constant $\beta$ and $y(v)$'s as a function of $v$ such that: for $v_{\min} \le v^* \le v_{\max}$,
\[
\textstyle \int_0^{v^*} v d y(v) - f(y(v^*)) \geq \tfrac{1}{\alpha} f^*(v^*) - \beta, 
\]
where $f(y) = 0$ for $0 \le y \le k$ and $+\infty$ for $y > k$, and $f^*(v) = kv$.
Differentiating both sides, we get that 
\begin{equation}
\label{eq:limsupllyinv}
\textstyle (v - f'(y)) dy \geq \tfrac{1}{\alpha} {f^*}'(v) dv, \textrm{ for } v_{\min} \le v \le v_{\max}.
\end{equation}
Since we have
$y(v) \le k$ for $v \le v_{\max}$, $v dy \geq \tfrac{k}{\alpha} dv$ for $v_{\min} \le v \le v_{\max}$.
Rearranging terms, we have $dy \geq \tfrac{k}{\alpha} \tfrac{1}{v} dv$ and thus,
\[ k \ge y(v_{\max}) - y(v_{\min}) \geq \tfrac{k}{\alpha} \ln \tfrac{v_{\max}}{v_{\min}} = \tfrac{k}{\alpha} \ln \rho.\]
So, the competitive ratio $\alpha$ is at least $\ln \rho$.

Next, we show that the competitive ratio is at least $\Omega(\tfrac{\log m}{\log\log m})$ even for $\rho = 1$.
Let us assume without loss of generality that $v_{\max} = v_{\min} = 1$.
Let $r = \log m$.
We will define $\log_r (m) + 1= \Theta(\tfrac{\log m}{\log\log m})$ different instances, and show that no online algorithm can be better than $\tfrac{1}{2}\log_r(m)$-competitive for all these instance.

For $0 \le i \le \log_r (m)$, instance $i$ is defined as follows:
let there be $i$ stages;
at stage $j$, $0 \le j \le i$, let there be $k r^j$ buyers with value $1$ for any bundle of size at least $m / r^j$ and value $0$ otherwise. That is, the buyers' value per item is $r^j / m$ at stage $j$.

The optimal offline solution for instance $i$ allocates all items to buyers at stage $i$, getting $r^i / m$ value per item and, thus, the optimal social welfare for instance $i$, $\opt(i)$, is $\opt(i) = k r^i$.

Assume for contradiction that there is a $\tfrac{1}{2} \log_r (m)$-competitive online algorithm.
By the same argument as in the proof of \pref{thm:fraclb}, it is characterized by the expected total number of items allocated up to stage $j$, denoted by $y(j)$. 
Recall that $y(j)$ is independent of $i$ other than $i \ge j$.
Note that it is also a function of $v_{\min}$ and $v_{\max}$ but we assume this implicitly and omit the parameters.
Let $\Delta y(j) = y(j) - y(j-1)$ denote the expected number of items allocated at stage $j$ (let $y(-1) = 0$).
The expected social welfare of the algorithm for instance $i$ is $\sum_{j = 0}^i \Delta y(j) \tfrac{r^i}{m}$.
Note that the contribution of stage $0$ to $i-1$ is at most $km \cdot \tfrac{r^{i-1}}{m} = \tfrac{1}{r} \opt(i) < \tfrac{1}{\log_r m} \opt(i)$, because there are $km$ items in total and item values are at most $\tfrac{r^{i-1}}{m}$ in these stages. 
So if the algorithm is $\tfrac{1}{2} \log_r m$ competitive, it must get at least $\tfrac{1}{\log_r (m)} \opt(i)$ social welfare from stage $i$ and, thus, $\Delta y(i) > \tfrac{km}{\log_r (m)}$.
Then, we have $y(\log_r (m) + 1) = \sum_{i = 0}^{\log_r(m)} \Delta y(i) > km$, contradicting the supply constraint.
\end{proof}

\end{document}